\documentclass[12pt]{article}

\usepackage[dvips]{graphicx}
\usepackage{amsmath, amssymb, amsthm}

\usepackage{bm}

\newcommand{\cX}{\mathcal{X}}
\newcommand{\cK}{\mathcal{K}}
\newcommand{\cP}{\mathcal{P}}
\newcommand{\bbR}{{\mathbb R}}
\newcommand{\payoff}{f}
\newcommand{\payoffc}{\phi}
\newcommand{\funcsp}{C_b^2}

\newtheorem{theorem}{Theorem}[section]
\newtheorem{prop}[theorem]{Proposition}

\newtheorem{rem}[theorem]{Remark}

\newtheorem{remark}[theorem]{Remark}

\author{Ryuichi Nakajima\footnote{Department of 
Mathematical Informatics, 
Graduate School of Information Science and Technology, 
University of Tokyo.},
Masayuki Kumon\thanks{masayuki\_kumon@smile.odn.ne.jp},
Akimichi Takemura\footnotemark[1]\\
and Kei Takeuchi\footnote{Emeritus, Graduate School of Economics, 
University of Tokyo}}
\date{Revised, June 2010}

\title{Approximations and asymptotics of 
upper hedging prices in multinomial models}

\begin{document}

\maketitle 

\begin{abstract}
We give an exposition and numerical studies 
of upper hedging prices in multinomial models from the viewpoint of 
linear programming and the game-theoretic probability of Shafer and Vovk.
We also show that, as the number of rounds goes to infinity,
the upper hedging price of a European option converges 
to the solution of the Black-Scholes-Barenblatt equation.
\end{abstract}



\section{Introduction} 
\label{sec:intro}

The Black-Scholes formula for the geometric Brownian motion model and
the Cox-Ross-Rubinstein formula for the binomial model are now treated
in many standard textbooks on mathematical finance 
(e.g.\ \cite{musiela-rutkowski-1st-ed,shreve-1,shreve-2}).  Since
these models are complete, the exact price for any contingent claim  is
determined by arbitrage argument.  On the other hand, incomplete
models such as the trinomial model are only briefly mentioned in
the textbooks because of difficulty associated with 
indeterminacy of prices of contingent claims.  

In fact only a few explicit results seem to be known on upper hedging prices
for the discrete-time multinomial models.  The purpose 
of this paper is to give an exposition of the exact and the asymptotic
behavior of upper hedging prices
of contingent claims in multinomial models.
We also show that the asymptotic upper hedging 
price of a European option is described by the Black-Scholes-Barenblatt equation
(e.g.\ \cite{avellaneda-buff},
\cite{vargiolu-2001}, 
\cite{gozzi-vargiolu}, 
\cite{meyer-2006},
Chapter 4 of \cite{pham}). 
The Black-Scholes-Barenblatt equation 
is usually considered for uncertain volatility models in continuous time
and its relation to multinomial models does not seem to be
stated in literature.

The advantage of discrete-time multinomial models is that
we can exactly compute the upper hedging price by backward induction
for moderate number of rounds
and various approximations to the upper hedging price can be compared to
the exact value.
%
Basic facts on the upper hedging price for discrete models
are well explained in Chapter 4 
of the first edition of 
Musiela and  Rutkowski \cite{musiela-rutkowski-1st-ed}.
General treatments of incomplete markets are given in Chapter 5 
of \cite{karatzas-shreve} and  \cite{schachermayer-notes}.
However they are concerned with continuous stochastic processes and
do not contain much numerical information on 
the behavior of upper hedging prices for discrete models.

An extensive numerical study of hedging in incomplete markets
is given in \cite{bertsimas-etal}.  Its authors consider a hedging strategy
which minimizes the  mean-squared error to the payoff
of a contingent claim under Markov-state dynamics.  As we see in Section
\ref{sec:formulation} below, for studying the behavior of upper hedging
prices we can not make convenient stochastic assumptions such as the
Markov property.  Results more relevant to the present paper have been 
given in \cite{ruschendorf} and \cite{courtois-denuit} by
convex ordering argument.
In particular for discrete time models with bounded
martingale differences, \cite{ruschendorf} proves that the
upper hedging price of a convex contingent claim is given by the  extremal
binomial model.  We reproduce this fact in Proposition \ref{prop:convex} below
by linear programming argument.

In this paper we use the framework of game-theoretic 
probability by Shafer and Vovk.
We prefer the framework because of the following reasons.
First, in game-theoretic probability only the protocol of a game between
``Investor'' and ``Market'' is formulated
without specification of a probability measure.  Therefore there is no need
to distinguish a risk neutral measure from an actual (or a physical) measure and
to consider the equivalence 
between them.  This
is advantageous because the extremal risk neutral measure corresponding 
to the upper hedging price usually has a support smaller than those in the interior
of the set of risk neutral measures. 
Second, some strong properties of a price path of Market can be proved in game-theoretic
probability 
without any stochastic assumption (e.g.\ see \cite{tkt-bernoulli},
\cite{vovk-rough-path}, \cite{vovk-fs} and references therein).  As a referee
pointed out \cite{bick-willinger-1994} studies non-probabilistic approach
for pricing in continuous time.
Third, the notion of upper hedging price is of central importance to game-theoretic
probability as shown in our recent works (\cite{tvs-aism}, \cite{svt-levy-zerone}) 
on game-theoretic probability.

The organization of the paper is as follows.  In Section
\ref{sec:formulation}  we give a linear programming formulation
of upper hedging prices in multinomial models and state some basic facts
in several propositions.  In Section \ref{sec:bounds} we give some simple
bounds for upper hedging prices.  Then in Section 
\ref{sec:pde}
we  show that, as the number of rounds goes to infinity,
the upper hedging price of a European option converges 
to the solution of an additive form of 
the Black-Scholes-Barenblatt equation.
In Section \ref{sec:examples} we present  numerical studies on the accuracy of
the partial differential equation and other approximations.
Some concluding remarks are given in Section \ref{sec:remarks}.

\section{Formulation of upper hedging price}
\label{sec:formulation}

In this section we formulate the upper hedging price for a multinomial game
from the viewpoint of linear programming and the game-theoretic
probability of Shafer and Vovk \cite{shafer-vovk}.  
Also we show some known facts on 
upper hedging prices.  

Let $\cX \subset \bbR$ be a finite set containing both
negative and positive elements.  The protocol of the
multinomial game of $N$ rounds with the initial capital of
$\cK_0=\alpha$ is written as follows.

\medskip
\newcommand{\indentone}{\hspace*{7mm}}
\newcommand{\indenttwo}{\hspace*{14mm}}
\noindent
\indentone ${\cal K}_0 =\alpha$\\
\indentone FOR $n=1,2,\dots,N$.\\
\indenttwo Investor announces $M_n\in{\mathbb R}$. \\
\indenttwo Market announces $x_n\in \cX$.\\
\indenttwo  ${\cal K}_n := {\cal K}_{n-1} + M_n x_n$.\\
\indentone END FOR
\bigskip

We call $\cX^N$ the {\em sample space} and 
$\xi=x_1\dots x_N \in \cX^N$ a {\em path} of Market's moves.  
For $1\le n\le N$, $\xi^n=x_1\dots x_n \in \cX^n$ is a partial path.
Investor's strategy $\cP$ is a function specifying
$M_n$ based on $\xi^{n-1}=x_1\dots x_{n-1}$:
\[
 \cP: x_1 \dots x_{n-1} \mapsto M_n.
\]
with some initial value $M_1=\cP(\Box)$, where $\Box$ denotes the initial  empty path.
When Investor adopts $\cP$, his capital at the end of round $n$ is written as
$\alpha + \cK_n^{\cP}$, where
\[
\cK_n^{\cP}= \sum_{i=1}^n \cP(\xi^{i-1})x_i.
\]

We can write the progression of the game in a rooted tree with the root $\Box$.
For $n < N$, $\xi^n$ is an intermediate node 
branching to $\xi^n x_{n+1}$, $x_{n+1}\in \cX$.  The final nodes
are $\xi\in \cX^N$.

We call a function $\payoff:\cX^N\rightarrow \bbR$ a {\em payoff} function or 
a {\it contingent claim}.
The upper hedging price (or simply the upper price) of $f$ is defined as
\[
\bar E(\payoff)=\inf \{ \alpha \mid \exists \cP \ \textrm{s.t.}\ \alpha + \cK_N^\cP (\xi) \ge f(\xi), \forall \xi\in \cX^N \}
\]
and the lower hedging price is  defined as
\begin{equation}
\label{eq:reciprocity}
\underline E (\payoff)= - \bar E(-\payoff).
\end{equation}
The upper hedging price and the lower hedging price are often called seller's price and 
buyer's price, respectively.
A strategy $\cP$ with the initial capital $\alpha$ satisfying
$\alpha + \cK_N^\cP (\xi) \ge f(\xi), \forall \xi\in \cX^N$, is
called a superreplicating strategy for $\payoff$.

The problem of obtaining the upper hedging price 
can be formulated in linear programming. Let $\cX=\{a_1, \dots, a_k\}$.
For the single step case $N=1$, $\bar E(f)$ is obtained as the following  minimum:
\begin{equation}
\label{eq:single-step}
\alpha =   \begin{pmatrix} 1 & 0 \end{pmatrix} \begin{pmatrix} \alpha \\ M_1 \end{pmatrix} \to \min \qquad
\textrm{s.t.} \ \begin{pmatrix} 1 & a_1 \\ 1 & a_2 \\ \vdots & \vdots \\1 & a_k \end{pmatrix} \begin{pmatrix} \alpha \\ M_1 \end{pmatrix} \ge \begin{pmatrix} 
\payoff(a_1) \\ \payoff(a_2) \\ \vdots \\ \payoff(a_k) \end{pmatrix}.
\end{equation}

For the two-step case $N=2$, Investor can choose his investment at round 2
depending on the Market's move in the first round. Therefore $\bar E(f)$ is written as
the following minimum:
\begin{align*}
&\alpha = \begin{pmatrix} 1 & 0 & 0 & 0 & \ldots &  0 \end{pmatrix} \begin{pmatrix} \alpha \\ M_1  \\ {M_{2}}_{|a_1} \\ M_{2|a_2} \\ \vdots \\ M_{2|a_k}  \end{pmatrix} \to \min \\
&\qquad
\textrm{s.t.\ } 
{\footnotesize \begin{pmatrix}1 & a_1 & a_1 & 0 & 0& \dots & 0\\
1 & a_1  & a_2 &0 &0 &\dots &0 \\ 
\vdots & \vdots & \vdots & \vdots & \vdots& \dots & \vdots\\
1 & a_1 & a_k & 0 & 0& \dots & 0 \\ 
1& a_2 & 0 & a_1 & 0 &\dots &0 \\ 
\vdots & \vdots & \vdots  & \vdots  & \vdots &\dots &\vdots \\
1& a_2 & 0 & a_k & 0 & \dots & 0 \\ 
\vdots & \vdots & \vdots  & \vdots  & \vdots &\dots &\vdots \\
1& a_k & 0 & 0 & 0 & \dots & a_k \\ 
\end{pmatrix} 
\begin{pmatrix} \alpha \\ M_1  \\ {M_{2}}_{|a_1} \\ M_{2|a_2} 
\\ \vdots \\ M_{2|a_k}\end{pmatrix} 
\ge \begin{pmatrix} \payoff(a_1a_1) \\ \payoff(a_1a_2) \\ \vdots \\\payoff(a_1a_k) 
\\ \payoff(a_2a_1) \\ \vdots \\\payoff(a_2a_k) \\ 
\vdots \\
\payoff(a_ka_k)  \end{pmatrix}. }
\end{align*}

For general $N$, the coefficient matrix $A_N: k^N \times (1 + (k^N - 1)/(k-1))$ 
is recursively defined as
\[A_N = \left(\begin{matrix} \bm{1}_{k^N} & \hat{A}_{N-1} \otimes 
\boldsymbol{1}_k
& \underbrace{ I_k  \otimes \cdots \otimes I_k}_{N-1} \otimes \begin{pmatrix}a_1\\ a_2\\ \vdots \\ a_k \end{pmatrix}\end{matrix}\right),
\]
where $\bm{1}_k$ is a $k$-dimensional vector of 1's, 
$\hat{A}_N$ is $A_N$ without  the first column and
$\otimes$ denotes the Kronecker product.  
Then 
\begin{equation}
\label{eq:multi-step-lp}
\bar E(f)=\min (1 \ 0 \ \dots \  0)\;\bm{m}, 
\quad \textrm{s.t.}\quad A_N \bm{m} \ge \bm{f},
\end{equation}
where $\bm{f}$ is a $k^N$-dimensional vector consisting of $f(\xi), \xi\in \cX^N$, and
\[
\bm{m}^\top=(\alpha\  M_1\  M_{2|a_1}\  \dots\  M_{2|a_k}\  \dots\  M_{n\mid \xi^{n-1}}\  \dots\ 
M_{N|a_k \dots a_k}).
\]

Since the size of
$A_N$ grows exponentially with $N$, it becomes difficult to directly solve
the linear programming problem for a general path-dependent contingent claim $\payoff$.
Exploiting the recursive structure of the coefficient matrix $A_N$, the
linear programming problem can be solved by backward induction, i.e.\
by solving the single-step optimizations for $N, N-1, \dots, 1$.
This will be explicitly described in \eqref{eq:backward-step} below.
Therefore the single-step optimization in
\eqref{eq:single-step} is essential.
However even with backward induction, for a general path-dependent $f$,
the number of  single-step optimizations grows exponentially with $N$.
This is because the number of nodes of the game tree grows exponentially
with $N$ and  the single-step optimization for the backward induction  is
performed at each node of the game tree.  
This difficulty is somewhat mitigated in the case of a European option $f$, where
$f$ is a function of $S_N=x_1 + \dots +x_N$ only.  Then the game tree can
be collapsed according to the values of $S_n = x_1 + \dots + x_n$ and
the number of nodes of the collapsed tree grows only polynomially with $N$.
In fact, for generic values of $a_1,\dots, a_k$, the
number of values taken by  $S_N$ is $\binom{N+k-1}{k-1}$.

Now we consider the single-step optimization in
\eqref{eq:single-step}.
Let $\cX$ contain $l$ negative elements and $m=k-l$ positive elements,
which are ordered as
\[
0 > a^-_1>a^-_2> \dots>a^-_l, \qquad 
0 \le  a^+_1<a^+_2< \dots<a^+_m. 
\]
Note that we allow the case $a_+^1=0$, which needs some special consideration.
For the single-step game the following result is given in Proposition 4.1.1 of
\cite{musiela-rutkowski-1st-ed}.  We give our own proof based 
on consideration of dual linear programming problem to 
\eqref{eq:single-step}.

\begin{prop}
\label{prop:single-step}
The upper hedging price of $\payoff$ in the single-step game $N=1$ is given by
\begin{equation}
\label{eq:one-step-upper}
\bar E(\payoff) = \max_{i,j} \left(\frac{a^+_j\payoff(a^-_i) -a^-_i\payoff(a^+_j)}{a^+_j - a^-_i}  \right)  .
\end{equation}
\end{prop}

\begin{proof}
Let $\bm{p}=(p^-_1\; \dots\; p^-_l\; p^+_1 \;\dots\; p^+_m)^\top$
and consider the following dual problem to \eqref{eq:single-step}: 
\begin{align}
&(\payoff(a^-_1)\; \dots\; \payoff(a^-_l)\;  \payoff(a^+_1) \; \dots \;\payoff(a^+_m) ) \;\bm{p} \to \max 
\nonumber \\
& \qquad \textrm{s.t.}  \quad 
\begin{pmatrix} 1 &  \cdots & 1 &1& \cdots & 1  \\ a^-_1 & \cdots & a^-_l
 & a^+_1 & \cdots & a^+_m \end{pmatrix} 
\bm{p}
 = \begin{pmatrix} 1 \\ 0  \end{pmatrix}, \quad \bm{p} \ge \bm{0}.
\label{eq:risk-neutral-single-step}
\end{align}
Note that the coefficient matrix is $2\times k$.  Therefore
the maximum is attained by a basic solution involving two variables.

First consider the case $a_1^+ > 0$.
If we choose two variables either from 
$\{p^-_1, p^-_2, \dots,p^-_l\}$ 
or from 
$\{p^+_1, p^+_2, \dots,p^+_m\}$, then the solution
does not satisfy $\bm{p} \ge \bm{0}$ and it is infeasible.
Therefore for a feasible basic solution we need to choose one variable $p^-_i$ from 
$\{p^-_1,p^-_2, \dots,p^-_l \}$ and another variable $p^+_j$ 
from  $\{p^+_1, p^+_2, \dots,p^+_m \}$.  Then the basic solution is given by
\begin{equation}
\label{eq:basic-solution}
p^-_i = p^-_{i;j} = \frac{a_j^+}{a_j^+ - a_i^-}, \ \  p^+_j = p^+_{j;i} = \frac{-a_i^+}{a_j^+ - a_i^+}
\end{equation}
and the value of the objective function is given by 
$(a^+_j\payoff(a^-_i) -a^-_i\payoff(a^+_j))/(a^+_j - a^-_i)$.

Now consider the case $a_1^+=0$.  Then for $j=1$ we have
$(a^+_j\payoff(a^-_i) -a^-_i\payoff(a^+_j))/(a^+_j - a^-_i)=f(0)$, which is
consistent with \eqref{eq:risk-neutral-single-step} and
\eqref{eq:basic-solution}.  Therefore this case does not require a separate statement.
This proves the proposition.
\end{proof}

By \eqref{eq:reciprocity} the lower hedging price $\underline E(f)$ for the single-step game is
given as
\[
\underline E(\payoff) = \min_{i,j} \left(\frac{a^+_j\payoff(a^-_i) -a^-_i\payoff(a^+_j)}{a^+_j - a^-_i}  \right)  .
\]

The vector $\bm{p}$ satisfying 
\eqref{eq:risk-neutral-single-step}  is a probability vector over $\cX$
such that its expectation is zero:
\[
\sum_{i=1}^l a^-_i p^-_i + \sum_{j=1}^m a^+_j p^+_j=0.
\]
Let $\cP_1$ denote the set of probability vectors satisfying (\ref{eq:risk-neutral-single-step}).
$\bm{p}\in \cP_1$  is called a risk neutral measure on $\cX$.
{}From  (\ref{eq:risk-neutral-single-step})  we have
$\bar E(\payoff)=\max_{\bm{p}\in \cP_1} E^{\bm{p}}[\payoff]$, 
where $E^{\bm{p}}[\cdot]$ denotes the usual expectation under $\bm{p}$.
This duality is well known in great generality (see Chapter 5 of \cite{karatzas-shreve}).
However the explicit expression of the upper hedging price in  (\ref{eq:one-step-upper})
is useful for the purpose of numerical backward  induction for $N \ge 2$.


\begin{rem}
\label{rem:multistep-expectation}
We can also consider the dual problem to (\ref{eq:multi-step-lp}) for the $N$-step problem.
It is easy to see that the set $\cP_N$ of risk neutral measures $\bm{p}=\{ p(\xi), \xi \in \cX^n \}$  is characterized as follows:
\begin{align}
\bm{p} \in \cP_N  \ &\Leftrightarrow \ 
p(x_1 \dots x_N)=p_1(x_1) \times p_2(x_2 |x_1)\times \cdots \times p_N(x_N | \xi^{N-1}),\label{eq:N-iff} \\
&\qquad {\rm s.t.}\ 
p_j(\cdot \mid  \xi^{j-1})\in \cP_1, \ j=1,\dots,n, \ \forall \xi^{j-1}\in \cX^{j-1}.
\nonumber
\end{align}
Also by duality we have $\bar E(\payoff)=\max_{\bm{p}\in \cP_N} E^{\bm{p}}[f]$.
Note that \eqref{eq:N-iff} holds even if 
$p(\xi)=0$ for some $\xi$ and some conditional probabilities on the right-hand side
are not defined.  At each node of the game tree, 
the conditional distribution of 
the extremal risk neutral measure corresponding to $\bar E(\payoff)$
is given by a basic solution of the form \eqref{eq:one-step-upper}.
\end{rem}

Based on (\ref{eq:one-step-upper}), the backward induction for obtaining $\bar E(\payoff)$
for the $N$-step game is described as follows.  
Define $\bar\payoff(\cdot, N-n): \cX^n \rightarrow \bbR$, $n=N, N-1, \dots, 0$, by 
\begin{equation}
\label{eq:backward-step}
\bar \payoff(\xi^n, N-n)= 
\max_{i,j} \left(\frac{a^+_j \bar\payoff(\xi^n a^-_i,N-n-1) 
  -a^-_i\bar\payoff(\xi^n a^+_j, N-n-1)}{a^+_j - a^-_i}  \right)
\end{equation}
with the initial condition $\bar\payoff(\xi^N,0)=\payoff(\xi^N)$, $\xi^N \in \cX^N$. 
Summarizing the arguments above we have the following proposition.

\begin{prop}
\label{prop:backward-induction}
The upper hedging price of $\payoff$ in the $N$-step game is given by
\[
\bar E(\payoff) = \bar\payoff(\Box,N).
\]
\end{prop}

Consider the case that the payoff $\payoff$ is defined on the whole $\bbR^N$.
Then we have the following result.

\begin{prop}
\label{prop:convex}
Suppose that $\payoff$ is convex.  Then 
$\bar E(\payoff)$  is given by the binomial model
concentrated on two outermost values $\{ a_l^-, a_m^+\}$.
For the case  $a_1^+>0$ $\underline E(f)$ 
is given by the binomial model concentrated 
on two innermost values $\{a_1^-, a_1^+\}$.
For the case $a_1^+=0$  $\underline E(f)=f(0)$.
\end{prop}

This result was proved by R\"uschendorf (\cite{ruschendorf}) 
by a convex ordering argument in a more general setting, 
where $\cX$ is a bounded interval.  See also \cite{courtois-denuit}.
Corresponding results for concave $\payoff$ are obtained by using
\eqref{eq:reciprocity}.
We now give a simple proof of this proposition using \eqref{eq:backward-step}.

\begin{proof}
First consider the one-step game $N=1$. By convexity of $\payoff$ it is easy to check
that for any $a_i^-$ and $a^+_j$ the following inequalities hold:
\[
\frac{a^+_m\payoff(a^-_l) -a^-_l\payoff(a^+_m)}{a^+_m - a^-_l}  \ge
\frac{a^+_j\payoff(a^-_i) -a^-_i\payoff(a^+_j)}{a^+_j - a^-_i}\\
\ge 
\begin{cases}
\frac{a^+_1\payoff(a^-_1) -a^-_1\payoff(a^+_1)}{a^+_1 - a^-_1} & \text{if}\  a_1^+ > 0, \\
\payoff(0) & \text{if}\ a_1^+=0.
\end{cases} 
\]
Hence the proposition holds for the single-step game. 

For $N>1$ we can use the induction.
We note that $\bar\payoff(\xi^n,N-n)$
is a convex combination of $f(\xi^n x_{n+1}\dots x_N)$, $x_{n+1}\dots x_N \in \cX^{N-n}$,
where the weights of the combination are given by the binomial model and they do
not depend on $\xi^n$.  In particular $\bar\payoff(\xi^n,N-n)$ is convex in $x_n$. 
\end{proof}

\begin{rem}
As pointed out in Remark 2 of \cite{ruschendorf}, the above result holds if $f$ is
component-wise convex in every $x_i$, $i=1,\dots,N$.
\end{rem}



\section{Some bounds for upper hedging prices}
\label{sec:bounds}

In this section we present some simple inequalities for
upper hedging prices.

We first investigate the relation between binomial model and
general multinomial model.
Pick a negative element $a_i^-$ and
a positive element $a_j^+$ from $\cX$ and restrict 
the move space $\cX$ of Market to the two element set 
$\cX_{i,j}=\{ a_i^-, a_j^+\}$.   Then the multinomial game is reduced to
a binomial model, where the price of any contingent claim is
given by an arbitrage argument.  
Let $E_{i,j}(\payoff)$ denote the price of $\payoff$ 
under the binomial model $\cX_{i,j}$.
For the two-element set $\cX_{i,j}$ there is no need to 
take the maximum in \eqref{eq:backward-step}.
On the other hand, for a multinomial model
$\cX$ we take the maximum in \eqref{eq:backward-step} at each node of the game.
This gives the following lower bound for $\bar E(\payoff)$:
\begin{equation}
\label{eq:ineq1}
\max_{i,j} E_{i,j}(\payoff) \le \bar E(\payoff).
\end{equation}
Similarly for the lower hedging price we have
\[
\min_{i,j} E_{i,j}(\payoff) \ge \underline E(\payoff).
\]

The above inequalities can be generalized by considering nested 
move spaces of Market.  Explicitly writing $\cX$ in the multinomial
game, let $\bar E(\payoff \mid \cX)$ and $\underline E(\payoff \mid \cX)$
denote the upper hedging price and the lower hedging price of
multinomial game with $\cX$ as Market's move space.
Consider two nested move spaces $\cX \subset \cX'$ of Market.
Then the same consideration as above gives the following inequalities:
\begin{equation}
\label{eq:nested-ineq}
\underline E(\payoff \mid \cX') \le
\underline E(\payoff \mid \cX) \le
\bar E(\payoff \mid \cX) \le
\bar E(\payoff \mid \cX').
\end{equation}
In Section 
\ref{sec:examples} we compare upper and lower hedging prices
in  trinomial and quadnomial models.

We can also consider dynamic restrictions of the move space of Market.
For example, we can consider 
the maximization in \eqref{eq:backward-step} only for even $n=2h$
and use the maximizer $i^*, j^*$ from this round 
for the subsequent round $n=2h+1$.  Or we can maximize e.g.\ 
every 10th step. 
By this dynamic restriction we again have a lower bound for the
upper hedging price.   As we discuss in the next section, 
this pruning of maximizations is conceptually close to discretization of
partial differential equation in \eqref{eq:main-pde}.

We now present another bound when $\payoff$ is a European option
depending only on $S_N$. Let $\payoff$ be defined on the whole 
$\bbR$.  Assume that $\payoff$ has the first derivative
$\payoff'$ which is of bounded variation.  Then $\payoff'$ is
written as a difference of two non-decreasing functions 
(Section 5.2 of \cite{royden}).
By taking the indefinite integral of $\payoff'$ we see that
$\payoff$ is written as a sum of 
a convex function and a concave function:
\[
\payoff = \payoff_1 + \payoff_2,  \qquad
\payoff_1: \text{convex}, \ 
\payoff_2: \text{concave}.
\]
By the subadditivity of the upper hedging price (Section 8.3 of
\cite{shafer-vovk}), we have
\begin{equation}
\label{eq:convex-concave-sum}
\bar E(\payoff) \le \bar E(\payoff_1) + \bar E(\payoff_2).
\end{equation}

The bounds in \eqref{eq:ineq1} and \eqref{eq:convex-concave-sum} are very simple.
Unfortunately as seen from numerical examples in Section
\ref{sec:examples} these bounds are generally not very tight.

\section{Limiting behavior of upper hedging price of an European option}
\label{sec:pde}

In this section we derive the limit of the upper hedging price
of a European option as $N\rightarrow\infty$ in an appropriate
sequence of games.  Let $\funcsp$ denote the space of functions $\bbR \rightarrow \bbR$
with compact support and continuous second derivatives.  
Let $f\in \funcsp$.
We consider a sequence of multinomial
games with $N$ rounds, where the payoff $\payoff_N$ for the $N$-th game is
given as
\begin{equation}
\label{eq:series-of-games}
\payoff_N(\xi^N)= \payoff(S_N/\sqrt{N}), \quad S_N = x_1 + \dots + x_N.
\end{equation}

Note that the expected value under the two-point distribution  in
\eqref{eq:basic-solution} is zero and the variance is given as
\[
  (a_i^-)^2 p^-_{i;j} + (a_j^+)^2 p_{j;i} = \frac{a_j^+ (a_i^-)^2 - a_i^- (a_j^+)^2}{a_j^+ - a_i^-}
= -a_i^- a_j^+.
\]
In view of this define the maximum variance and the minimum variance of $\cP_1$ as
\[
\bar \sigma^2 = - a_l^- a_m^+ \ \ > \ \  \underline \sigma^2 = -a_1^- a_1^+.
\]
We now state the following theorem.

\begin{theorem} 
\label{thm:bsb}
Let $\payoff\in \funcsp$ and let  $\payoff_N$ be defined by
(\ref{eq:series-of-games}).  Assume $\underline \sigma^2 > 0$. Then
\label{thm:main}
\[\lim_{N\rightarrow\infty}\bar E(\payoff_N)= \bar \payoffc(0,1),
\]
where $\bar \payoffc(s,t)$, $s\in \bbR$, $0\le t\le 1$, 
satisfies the following partial differential equation
\begin{equation}
\label{eq:main-pde}
\frac{\partial}{\partial t} \bar \payoffc(s,t)= \frac{\tilde \sigma^2}{2}
\frac{\partial^2}{\partial s^2} \bar\payoffc(s,t), \qquad 
\begin{cases}
\tilde \sigma^2=\bar\sigma^2, &  {\rm if}\  
\frac{\partial^2}{\partial s^2} \bar\payoffc(s,t)\ge 0 ,\\
\tilde \sigma^2=\underline\sigma^2,  &  {\rm if}\ 
\frac{\partial^2}{\partial s^2} \bar\payoffc(s,t)< 0,
\end{cases}
\end{equation}
with the boundary condition $\bar\payoffc(s,0)=\payoff(s), s\in \bbR$.
\end{theorem}

Similarly the following partial differential equation describes the
limiting lower price of $\payoff_N$. 
\begin{equation}
\label{eq:main-pde2}
\frac{\partial}{\partial t} \underline \payoffc(s,t)= \frac{\tilde \sigma^2}{2}
\frac{\partial^2}{\partial s^2} \underline\payoffc(s,t), \qquad 
\begin{cases}
\tilde \sigma^2=\underline\sigma^2, &  {\rm if}\  
\frac{\partial^2}{\partial s^2} \underline\payoffc(s,t)\ge 0 ,\\
\tilde \sigma^2=\bar\sigma^2,  &  {\rm if}\ 
\frac{\partial^2}{\partial s^2} \underline\payoffc(s,t)< 0.
\end{cases}
\end{equation}
We can understand \eqref{eq:main-pde} 
as a piecewise heat equation, where the diffusion coefficient
depends on the convexity or concavity of $\bar\payoffc$.
As pointed out by a referee the equation \eqref{eq:main-pde}
is studied in \cite{ping-shige}.

We stated Theorem \ref{thm:bsb} for the
simple setting of $\underline\sigma^2 > 0$ and $f\in \funcsp$.
Theorem 4.6.9 of \cite{pham} states that the Black-Scholes-Barenblatt
equation holds for a payoff function with linear growth condition:
for some $a,b>0$, 
$|f(s)| \le a + b |s|$, $\forall s\in {\mathbb R}$. 
In view of 
this result we expect our Theorem \ref{thm:bsb}
also holds for continuous 
$f$ satisfying a linear growth condition.  However justifying 
the limiting argument from discrete time to continuous time does not
seem to be simple.

\begin{remark}
\label{rem:viscosity}
The case $\underline\sigma^2=0$ needs a special consideration, although 
Theorem \ref{thm:bsb} still holds for this case.
In view of Theorem 4.6.9 of \cite{pham}, 
the notion of viscosity solution (cf.\  \cite{users-guide}) is needed for  (\ref{eq:main-pde}).

In Section 6.3 of \cite{shafer-vovk} this case was treated using parabolic potential
theory.  The equivalence of (\ref{eq:main-pde}) to the treatment in Section 6.3 of  \cite{shafer-vovk}
is seen by the following intuitive argument.  
If $\underline\sigma^2=0$, then $(\partial/\partial t) \bar \phi(s,t)\ge 0$,
$\forall s,t$, and $\bar\phi$ is non-decreasing in $t$. $\bar\phi$ strictly increases in $t$
at some $(s_0,t_0)$ if and only if $\bar\phi(s,t)$ is strictly convex in $s$ at this point.  Then for all
$t\ge t_0$, $\bar\phi(s,t)$ is (at least weakly) convex in $s$. This implies that
$\tilde\sigma^2 = \bar\sigma^2$ if and only if $\bar\phi(s,t) > f(s)$, which corresponds
to the ``continuous region'' in Section 6.3 of \cite{shafer-vovk}.

The numerical behavior of (\ref{eq:main-pde}) for this case is well illustrated in 
Figure 6.2 of \cite{shafer-vovk}.  It should also be noted that
$\bar\phi(s,\infty)=\lim_{t\rightarrow\infty}\bar\phi(s,t)$  is the least
concave majorant (concave envelope) of $f$.
\end{remark}

Note that \eqref{eq:main-pde} is an additive form of 
the Black-Scholes-Barenblatt equation in which  the right-hand side
of \eqref{eq:main-pde} multiplied by $s^2$:
\begin{equation}
\label{eq:BSB-orig}
\frac{\partial}{\partial t} \bar \payoffc =  \frac{\tilde \sigma^2}{2}
s^2 \, \frac{\partial^2}{\partial s^2} \bar\payoffc, \qquad 
\begin{cases}
\tilde \sigma^2 = \bar\sigma^2, &  {\rm if}\  
\frac{\partial^2}{\partial s^2} \bar\payoffc\ge 0 ,\\
\tilde \sigma^2=\underline\sigma^2,  &  {\rm if}\ 
\frac{\partial^2 }{\partial s^2} \bar\payoffc< 0.
\end{cases}
\end{equation}
Consider a multiplicative model, where Reality chooses positive $x_n$'s and
$S_N = x_1 \times \dots \times x_N$ is the product of $x_n$'s.  A European
option is of the form $\payoff(S_N)$. 
In \eqref{eq:multiplicative-sn} below we see that the resulting partial differential equation
is exactly the Black-Scholes-Barenblatt equation.
Note that the multiplicative model is standard in finance literature, although
it is well known that the pioneering work of Bachelier (\cite{bachelier})
was formulated
in the additive form. In this paper we use additive model, because
game-theoretic protocols are usually formulated in an additive form and
also because the limiting partial differential equation is a more
direct generalization of the heat equation.



A rigorous proof of our theorem is somewhat tedious 
and we first give some heuristic arguments as to why 
\eqref{eq:main-pde} should hold.  Later we give a more formal proof, by
considering an approximate superreplicating strategy as in Section 6.2 of
\cite{shafer-vovk}.

For our argument it is more convenient to rescale the move space of Market in the
$N$-th game to
\begin{equation}
\label{eq:N-th-game}
\cX_N = \{ \frac{a_1^-}{\sqrt{N}}, \dots, \frac{a_l^-}{\sqrt{N}}, 
    \frac{a_1^+}{\sqrt{N}}, \dots, \frac{a_m^-}{\sqrt{N}} \}.
\end{equation}
After this rescaling,
the backward induction in \eqref{eq:backward-step}  for the $N$-th game is written as
\begin{equation}
\label{eq:backward-sn}
\bar \payoff_N(s , N-n)= 
\max_{i,j} \big(p_{i;j}^- \bar\payoff_N(s + \frac{a^-_i}{\sqrt{N}},N-n-1)  + p_{j;i}^+
\bar\payoff_N(s + \frac{a_j^+}{\sqrt{N}},N-n-1)\big),
\end{equation}
where $s=S_n$ and
$p_{i;j}^-, p_{j;i}^+$ are given by \eqref{eq:basic-solution}.
The initial condition is given by  $\bar\payoff_N(S_N,0)=\payoff(S_N)$.
Note that by backward induction (\ref{eq:backward-sn}) defines $\bar \payoff_N(s,N-n)$ 
for all $s\in \bbR$, since $\payoff$ is defined on the whole $\bbR^1$.

As in the proof of 
Proposition \ref{prop:convex}
$\bar\payoff_N(s,N-n)$ is
a convex combination of values $\payoff(s+S_{N-n})$.
It should be noted that, 
unlike the case of convex $\payoff$ in Proposition \ref{prop:convex},
the weights of the convex
combination depend on $s$.    However as seen from the proof of
Proposition \ref{prop:convex}, the weights are
concentrated either on the two outermost values $\{ a_l^-, a_m^+\}$
or on the two innermost values $\{a_1^-, a_1^+\}$, 
depending on the convexity of $\bar\payoff_N(s,N-n)$ in $s$.
Hence in each interval of convexity or concavity of $\bar\payoff_N(s,N-n)$,
it is twice continuously differentiable in $s$. 
In our numerical studies we found that if the payoff function $f$ is smooth
and has only finite number of inflection points, then
$\bar \payoff_N(s,N-n)$ as a function of $s$ has no more
inflection points than $\payoff$.

Write $\nu=N-n-1$.
Then each term in the right-hand side of \eqref{eq:backward-sn} is expanded as
\begin{align*}
&p_{i;j}^- \bar\payoff_N(s + \frac{a^-_i}{\sqrt{N}},\nu)  + p_{j;i}^+
\bar\payoff_N(s + \frac{a_j^+}{\sqrt{N}},\nu) \\
& \qquad =
p_{i;j}^- \Big(\bar\payoff_N(s,\nu) + \frac{a^-_i}{\sqrt{N}} \bar\payoff_N'(s,\nu) + \frac{(a^-_i)^2}{2N}
\bar\payoff_N''(s+\theta_i\frac{a_i^-}{\sqrt{N}},\nu) \Big)
\\
&\qquad\quad   + 
p_{j;i}^+ \Big(\bar\payoff_N(s,\nu) + \frac{a^+_j}{\sqrt{N}} \bar\payoff_N'(s,\nu) + \frac{(a^+_j)^2}{2N}\bar\payoff_N''(s+\theta_j \frac{a_j^+}{\sqrt{N}},\nu) 
\Big)\\
&\qquad = \bar \payoff_N(s,\nu) + \frac{-a_i^- a_j^+}{2N}  \bar f_N''(s,\nu) + R_N,
\qquad (0< \theta_i, \theta_j < 1),
\end{align*}
where derivatives are with respect $s$  and
$|NR_N|=o(1)$ uniformly in $s$ and $\nu$.
Then \eqref{eq:backward-sn} is written as
\begin{equation}
\label{eq:difference-equation}
N (\bar\payoff_N(s,\nu+1)- \bar \payoff_N(s,\nu))=
\max_{i,j}(-a_i^- a_j^+ \frac{1}{2} \bar f_N''(s,\nu) + NR_N).
\end{equation}
If we ignore $NR_N=o(1)$, the right hand side is maximized
by $(i,j)=(l,m)$ or $(i,j)=(1,1)$ depending on the sign of
$\bar f_N''(s,\nu)$.

Now by rescaling time axis define
\[
\bar\payoffc_N(s,t)=\bar \payoff_N(s, Nt), \quad s\in \bbR, \ t\in [0,1].
\]
Then \eqref{eq:difference-equation} is written as
\begin{equation}
\label{eq:backward2}
\frac{\bar\payoffc_N(s,t+\Delta t)- \bar\payoffc_N(s,t)}{\Delta t}
= \max_{i,j}(-a_i^- a_j^+ \frac{1}{2}\bar \payoffc_N''(s,t) + NR_N), \qquad \Delta t = 1/N.
\end{equation}
This clearly corresponds to \eqref{eq:main-pde}.
However it seems
difficult to let $N\rightarrow\infty$ in 
\eqref{eq:backward2} and prove our theorem directly,
although the finite difference approximation
to HJB equations in Chapter IX of \cite{fleming-soner}
should hold in some form.

At this point we indicate how the Black-Scholes-Barenblatt equation \eqref{eq:BSB-orig}
arises in the multiplicative case.  In the multiplicative multinomial
model, $x_n$ is assumed to be of the form 
$ x_n-1 \in  {\cal X}_N
$,
where ${\cal X}_N$ is given in \eqref{eq:N-th-game}.
Let $S_n=x_1 \times \dots \times x_n$. Then
\[
S_{n+1}=S_n \times x_{n+1} =S_n + S_n \times (x_{n+1}-1).
\]
The expansion of the right-hand side of 
\eqref{eq:backward-sn} in the multiplicative model is
\begin{equation}
\label{eq:multiplicative-sn}
p_{i;j}^- \bar\payoff_N(s + s\frac{a^-_i}{\sqrt{N}},\nu)  + p_{j;i}^+
\bar\payoff_N(s + s\frac{a_j^+}{\sqrt{N}},\nu)
= \bar \payoff_N(s,\nu) + \frac{-a_i^- a_j^+}{2N} s^2\bar f_N''(s,\nu) + R_N.
\end{equation}
This corresponds to \eqref{eq:BSB-orig}.

Instead of the above direct approach, 
knowing that  \eqref{eq:main-pde} should hold, we can construct
an approximate superreplicating strategy of Investor and prove our theorem
as in Section 6.2 of \cite{shafer-vovk}.  In the following proof, 
in order to show the inequality  $\bar \payoffc(0,1) \ge \limsup_N\bar E(f_N)$ we 
adopt a suggestion by a referee.

\begin{proof}[Proof of Theorem \ref{thm:main}]
By Theorem 4.6.9 of \cite{pham}, 
the solution $\bar\phi$ to  \eqref{eq:main-pde} 
has a continuous first-order derivative in $t$ and
a continuous second-order derivative in $s$. 
See also Theorem 11 of \cite{gozzi-vargiolu} and \cite{vargiolu-2001}.
Consider the following sequence
\[
\bar\payoffc(0,1), \bar\payoffc(S_1, \frac{N-1}{N}), \dots,
\bar\payoffc(S_{N-1},\frac{1}{N}),\bar\payoffc(S_N,0),
\]
where $\bar\payoffc(S_N,0)=\payoff(S_N)$.
Writing $dS_n = S_{n+1}-S_n = x_{n+1}$, $D_n=1-n/N$, $dD_n=-1/N$, we can expand the successive difference as
\begin{align}
d \bar\payoffc(S_n, D_n)&=\bar\payoffc(S_{n+1},D_{n+1})-\bar\payoffc(S_n, D_n)
\nonumber \\
&=\frac{\partial}{\partial s} \bar\payoffc(S_n, D_n) dS_n +  \frac{1}{2} \frac{\partial^2}{\partial s^2}\bar\payoffc(S_n, D_n) (dS_n)^2 + \frac{\partial}{\partial t}\bar\payoffc(S_n, D_n) dD_n
+ R_N \nonumber\\
&=\frac{\partial}{\partial s} \bar\payoffc(S_n, D_n) dS_n -  \frac{1}{2N}
\big(\tilde\sigma^2(S_n, D_n) - N (dS_n)^2) \frac{\partial^2}{\partial s^2}\bar\payoffc(S_n, D_n)
+ R_N,
\label{eq:sv-proof}
\end{align}
where $N R_N= o(1)$ is uniformly  
in $s$ and $t$.
Consider a Markov superhedging strategy $\cP$ 
(cf.\ \cite{gozzi-vargiolu}) of Investor 
which chooses $M_n=(\partial/\partial s) \bar\payoffc(S_n, D_n)$.
By adding \eqref{eq:sv-proof} for $n=1,\dots,N$ we have
\begin{equation}
\label{eq:expansion-sv}
\payoff(S_N) - \bar\payoffc(0,1)
= \cK_N^{\cP} - 
 \frac{1}{2N}\sum_{n=1}^N 
\big(\tilde\sigma^2(S_n, D_n) - N (dS_n)^2) \frac{\partial^2}{\partial s^2}\bar\payoffc(S_n, D_n)
+ o(1).
\end{equation}
At this point we adopt a suggestion by a referee.  From Proposition 
\ref{prop:single-step} and Remark \ref{rem:multistep-expectation} we know
that the upper hedging price is computed as the expected value of 
$\payoff(S_N)$ under the extremal risk neutral measure, say ${\bm p}^*={\bm p}^*_N$. Under 
any risk neutral measure, $\cK_n^{\cP}$, $n=1,\dots,N$, is a measure-theoretic martingale
and its expected value is zero.  Now consider the expected value of 
\[
\big(\tilde\sigma^2(S_n, D_n) - N (dS_n)^2) \frac{\partial^2}{\partial s^2}\bar\payoffc(S_n, D_n)
\]
under ${\bm p}^*$.  We can evaluate the expected value, first by
conditioning on $x_1, \dots, x_n$.  By the definition of $\tilde \sigma^2$,
under ${\bm p}^*$  the conditional variance of $\sqrt{N}dS_n$ satisfies
\begin{align*}
E (N (dS_n)^2 \mid x_1, \dots, x_n) &\le \tilde\sigma^2 \quad \text{if}\quad \frac{\partial^2}{\partial s^2}\bar\payoffc(S_n, D_n)\ge 0\\
E (N (dS_n)^2 \mid x_1, \dots, x_n) &\ge  \tilde\sigma^2 \quad \text{otherwise}.
\end{align*}
Therefore taking the unconditional expected value of
\eqref{eq:expansion-sv}
under ${\bm p}^*$ we have
$\bar E(\payoff_N)  - \bar\payoffc(0,1)\le 0$ except for a term of order $o(1)$.
Hence $\limsup_N \bar E(\payoff_N) \le \bar\payoffc(0,1)$.

Conversely, consider 
Market's randomized moves chosen according to the
extremal risk neutral measure corresponding to $\bar E(f_N)$, which is
concentrated to two outermost values $\{a_l^-, a_m^+\}$ or two innermost
values $\{a_1^-, a_1^+\}$ at each node of the game tree,  depending
on the sign of $\bar\payoffc''$ 
(see Remark \ref{rem:multistep-expectation}).
Investor's capital is a measure-theoretic martingale under this 
risk neutral measure.  Since the measure is supported on two
points at each node of the game tree, we can modify the 
standard argument for binomial models to show that
the expected value of the payoff $\payoff$ converges to 
$\bar\payoffc(0,1)$ under the measure.  On the other hand
$\bar E(f_N)$ is the supremum over all possible moves of Market.  
Therefore we have
$\bar\payoffc(0,1) \le \liminf_N \bar E(f_N)$.
\end{proof}

\begin{remark}
In the above proof we partly used measure theoretic arguments as suggested by a
referee.  Although we can give a purely game theoretic proof in the line of 
Section 6.2 of \cite{shafer-vovk}, it is somewhat tedious.
The difficulty lies in the fact that
$\frac{\partial^2}{\partial s^2}\bar\payoffc(S_n, D_n)$ is path-dependent.
Note that 
by  the 
game-theoretic law of large numbers (\cite{shafer-vovk}, \cite{kumon-takemura-2008}),
Investor can force that $S_N/\sqrt{N}$ converge to 0.  
This implies that for large $N$ the empirical distribution of
Market's moves is approximately a risk neutral measure 
and $\sum_{n=1}^N (dS_n)^2$ is the
variance of a risk neutral measure.  However because each 
$(dS_n)^2$ is multiplied by $\frac{\partial^2}{\partial s^2}\bar\payoffc(S_n, D_n)$,
the convergence $S_N/\sqrt{N}\rightarrow 0$ does not imply
 \[
 \liminf_N \frac{1}{N} \sum_{n=1}^N \big(\tilde\sigma^2(S_n, D_n) - N (dS_n)^2) \frac{\partial^2}{\partial s^2}\bar\payoffc(S_n, D_n) \ge 0.
 \]
Although the argument can be fixed by discretization of the values of 
$\frac{\partial^2}{\partial s^2}\bar\payoffc(S_n, D_n)$, we omit the details.
\end{remark}


Numerically 
\eqref{eq:main-pde} can be solved by the following backward induction:
1) discretization of the interval $[0,1]$ 
and $\bbR$, 
2) approximation of the second derivative $(\partial^2/\partial s^2) \bar\payoffc$
by the second order difference of three neighboring points.
Actually this backward induction is entirely similar to the
exact backward induction in 
\eqref{eq:backward-step}.  When the discretization is not fine enough, then
the above numerical approximation corresponds to pruning of maximizations
discussed in Section \ref{sec:bounds}.  This suggests that a coarse discretization
of the partial differential equation yields an approximation which is less than
the the true $\bar\payoffc(s,0)$.

\section{Numerical examples}
\label{sec:examples}

In this section we check results of this paper by numerical computation.

We first calculate the upper hedging price and the lower hedging price of Butterfly spread option 
$\payoff (S_N)=\max(0,S_N + 0.5) - 2 \max(0,S_N - 0.5) + \max(0,S_N - 1.5)$
in Figure \ref{fig:Butterfly} 
under the trinomial model $(a_1=-1/\sqrt{N}, a_2 = 1/\sqrt{N}, a_3 = 2/\sqrt{N})$. 
Although Butterfly spread does not satisfy the differentiability condition of 
Theorem \ref{thm:main}, it can be arbitrarily closely approximated by
a payoff function satisfying the condition of Theorem \ref{thm:main}.
The results are shown in Figure \ref{fig:Butterfly2} in conjunction with the price under the binomial models. 
From Figure \ref{fig:Butterfly2}, we see that the upper price and the lower price are different from the price under the binomial models.

\begin{figure}[htbp]
  \centering 
  \includegraphics[width=5cm,keepaspectratio,clip]{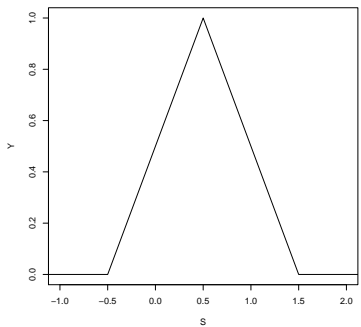}
  \caption{Butterfly spread}
  \label{fig:Butterfly}
\end{figure}
\begin{figure}[htbp]
  \centering
  \includegraphics[width=8cm,keepaspectratio,clip]{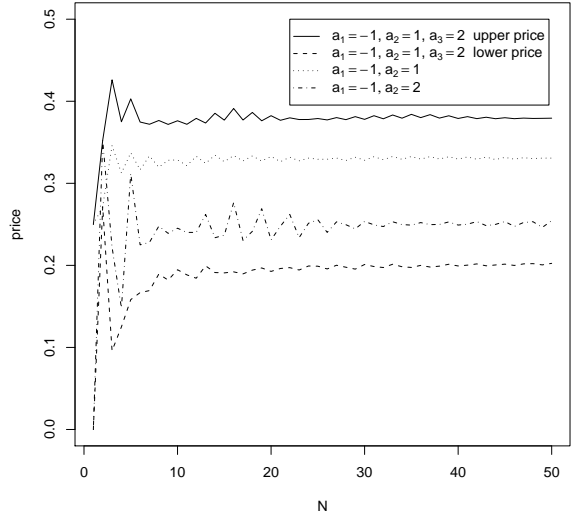}
  \caption{Comparison of binomial model and trinomial model}
  \label{fig:Butterfly2}
\end{figure}

We now add a new Market's move $a_4$ to the trinomial model and
compare the former trinomial model $(a_1=-1, a_2 = 1, a_3 = 2)$ to this quadnomial model. 
We consider the following three values of $a_4$ as depicted in Figure \ref{fig:Quad}.
\begin{enumerate}
\item $a_4=\frac{2.5}{\sqrt{N}}\ \ (a_1 < 0 < a_2 <a_3 <a_4)$ \qquad ``outside''.
\item $a_4=\frac{1.5}{\sqrt{N}}\ \ (a_1 < 0 < a_2 <a_4 <a_3)$ \qquad ``middle''.
\item $a_4=\frac{0.5}{\sqrt{N}}\ \ (a_1 < 0 < a_4 <a_2 <a_3)$ \qquad ``inside''.
\end{enumerate}

\begin{figure}[htbp]
  \centering 
  \includegraphics[width=10cm,keepaspectratio,clip]{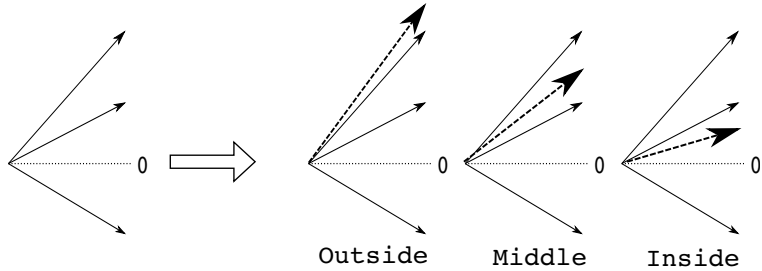}
  \caption{Expansion trinomial model into quadnomial model}
  \label{fig:Quad}
\end{figure}

Figures \ref{fig:Quad1}, \ref{fig:Quad2}, \ref{fig:Quad3} show the upper hedging prices of 
the butterfly spread for these quadnomial models compared to those of the trinomial model.
In Figure \ref{fig:Quad2} the upper hedging prices under the quadnomial model equal those  under the trinomial model with increasing $N$, whereas in Figures \ref{fig:Quad1}, \ref{fig:Quad3} the upper hedging prices under the quadnomial models differ from those under the trinomial model. 

\begin{figure}[htbp]
\begin{minipage}{.33\linewidth}
\begin{center}
\includegraphics[width=5.7cm]{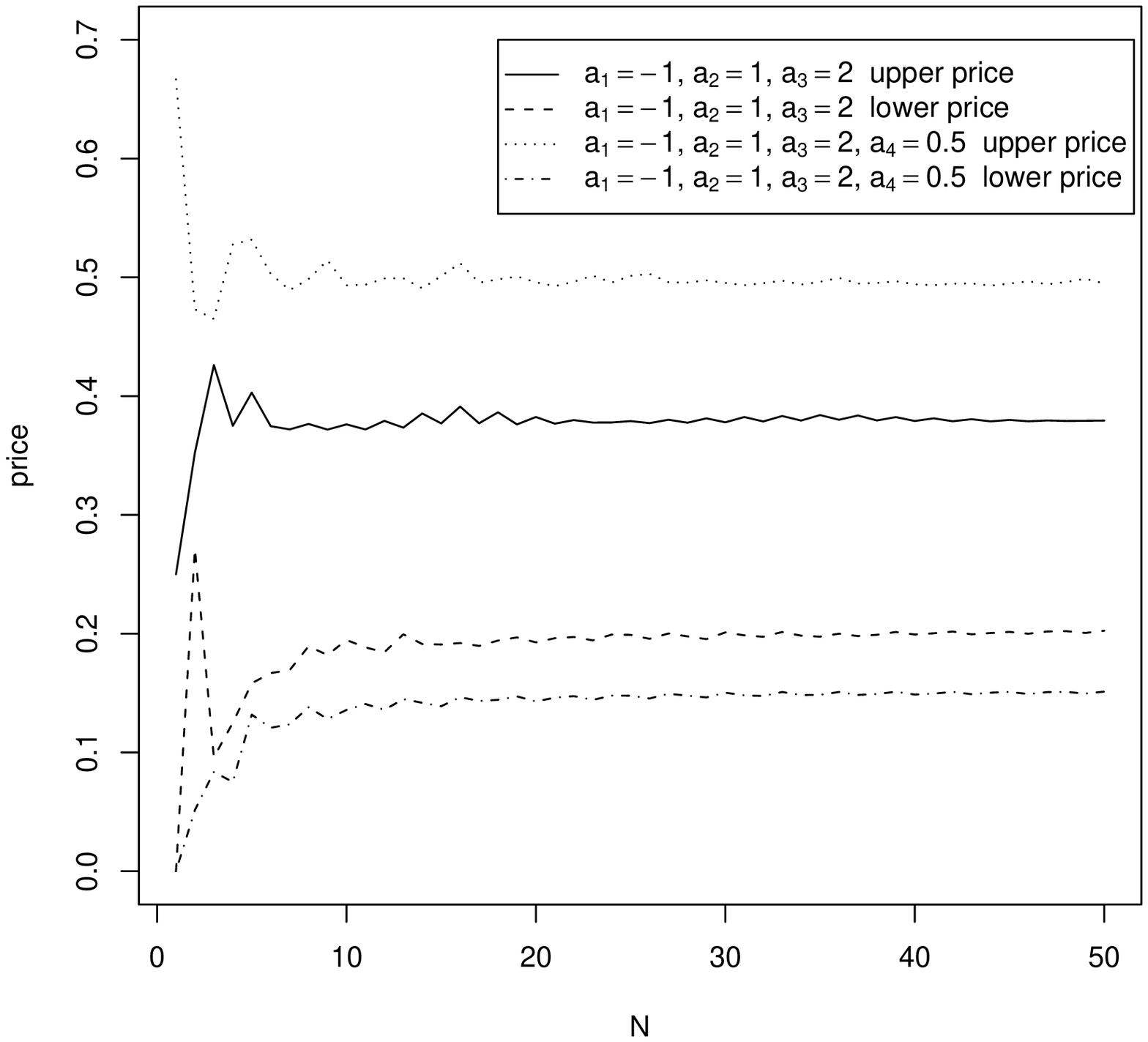}
 \caption{ $(\frac{2.5}{\sqrt{N}}=a_4 > a_3)$}
  \label{fig:Quad1}
\end{center}
\end{minipage}
\begin{minipage}{.33\linewidth}
\begin{center}
  \includegraphics[width=5.7cm]{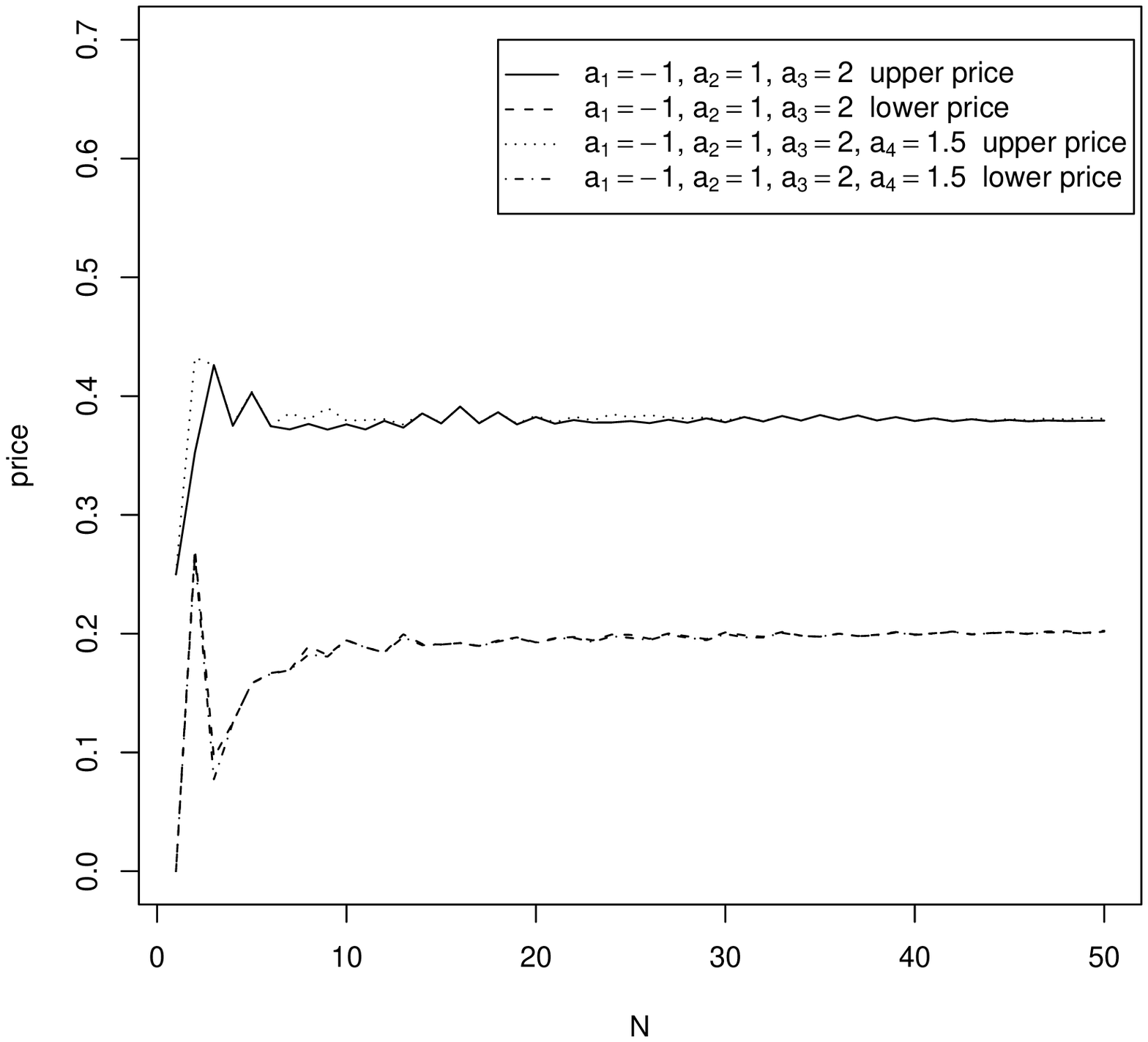}
  \caption{\footnotesize $(a_2 < \frac{1.5}{\sqrt{N}}=a_4 < a_3)$}
  \label{fig:Quad2}
\end{center}
\end{minipage}\nolinebreak
\begin{minipage}{.37\linewidth}
\begin{center}
  \includegraphics[width=5.7cm]{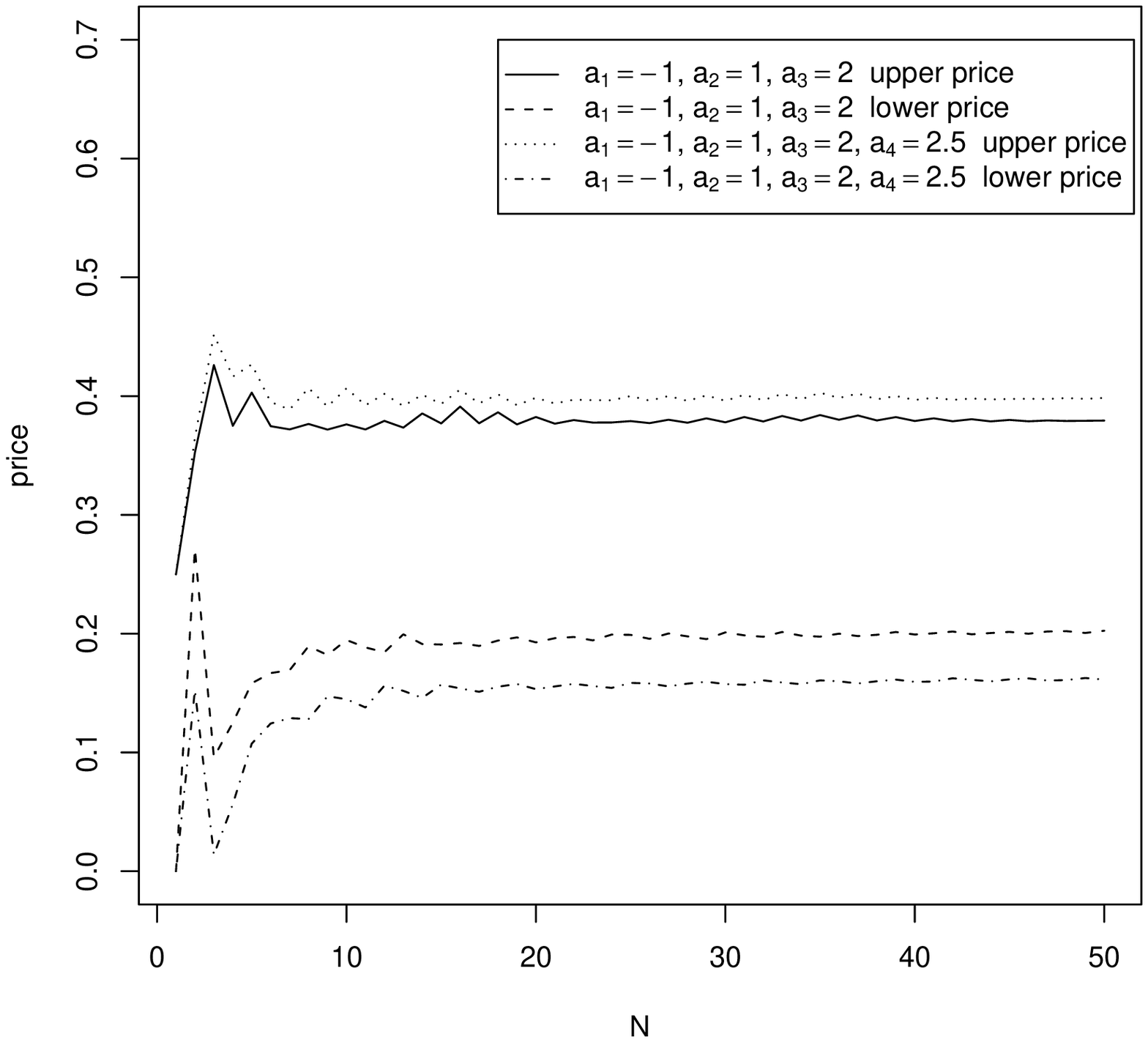}
  \caption{$(\frac{0.5}{\sqrt{N}}=a_4 < a_2)$}
  \label{fig:Quad3}
\end{center}
\end{minipage}
\end{figure}

Next we consider the payoff $\payoff = \sin (10S_n)$, which has a lot of changes from convexity to concavity,  and similarly calculate the upper hedging prices.
The results are shown in Figures \ref{fig:Quad4}, \ref{fig:Quad5}, \ref{fig:Quad6}. 
Also in this case the upper hedging prices under the quadnomial model equal those under the trinomial model with increasing $N$, provided that $a_4 = 1.5$.

\begin{figure}[htbp]
\begin{minipage}{.33\linewidth}
\begin{center}
  \includegraphics[width=5.7cm]{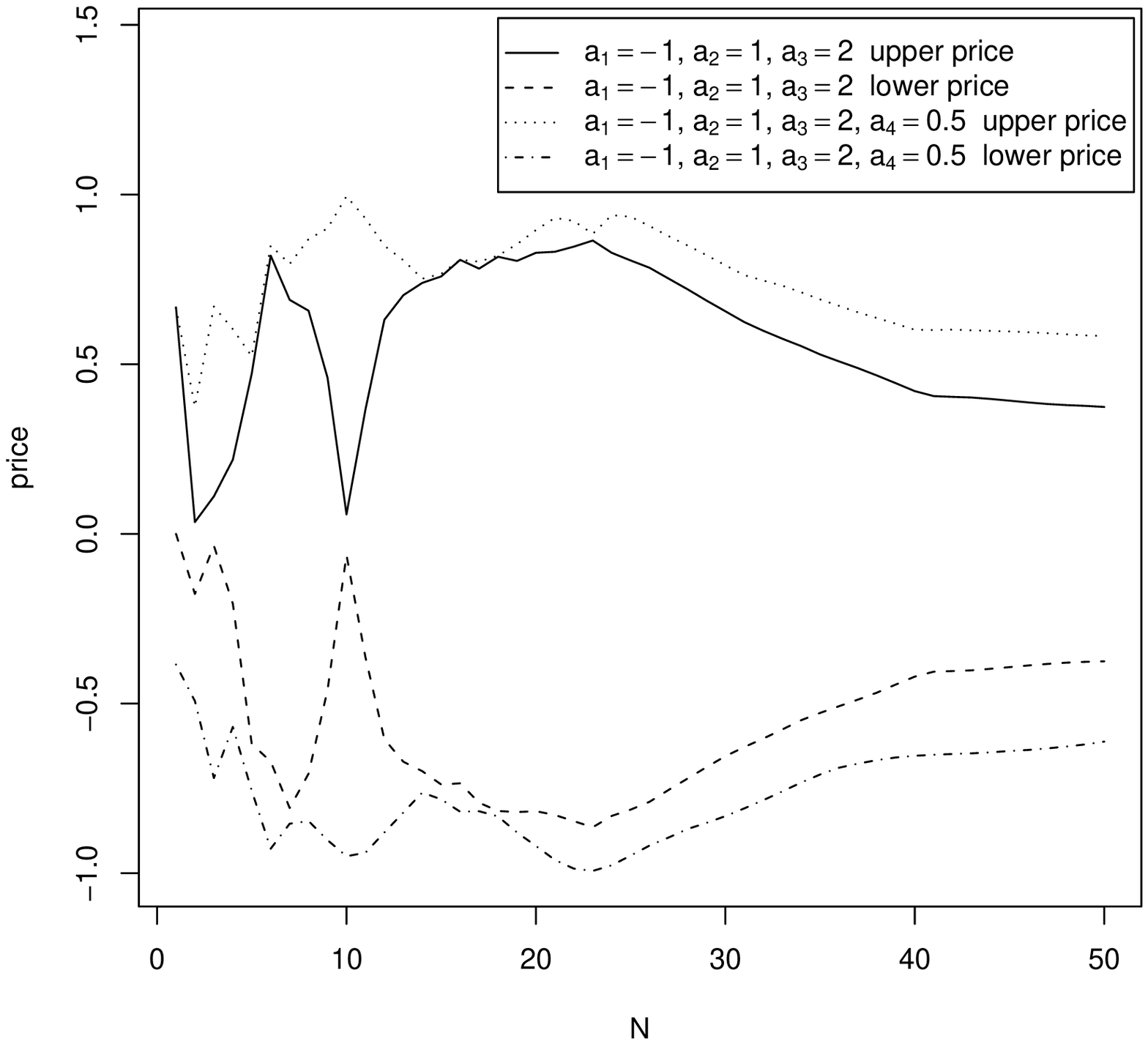}
  \caption{$(\frac{2.5}{\sqrt{N}}=a_4 > a_3)$}
  \label{fig:Quad4}
\end{center}
\end{minipage}
\begin{minipage}{.33\linewidth}
\begin{center}
  \includegraphics[width=5.7cm]{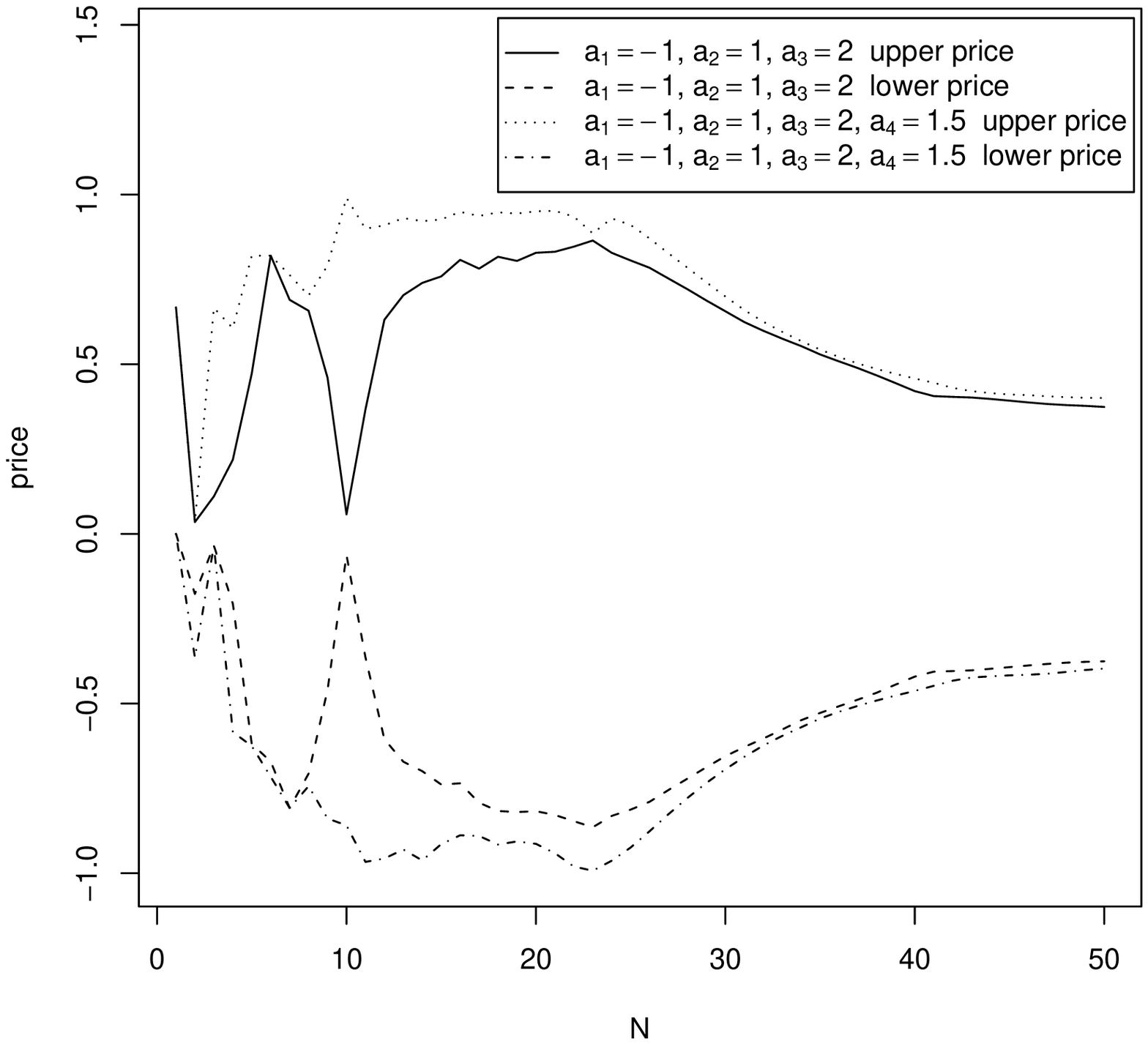}
  \caption{\footnotesize $(a_2 < \frac{1.5}{\sqrt{N}}=a_4<a_3)$}  
  \label{fig:Quad5}
\end{center}
\end{minipage}\nolinebreak
\begin{minipage}{.37\linewidth}
\begin{center}
  \includegraphics[width=5.7cm]{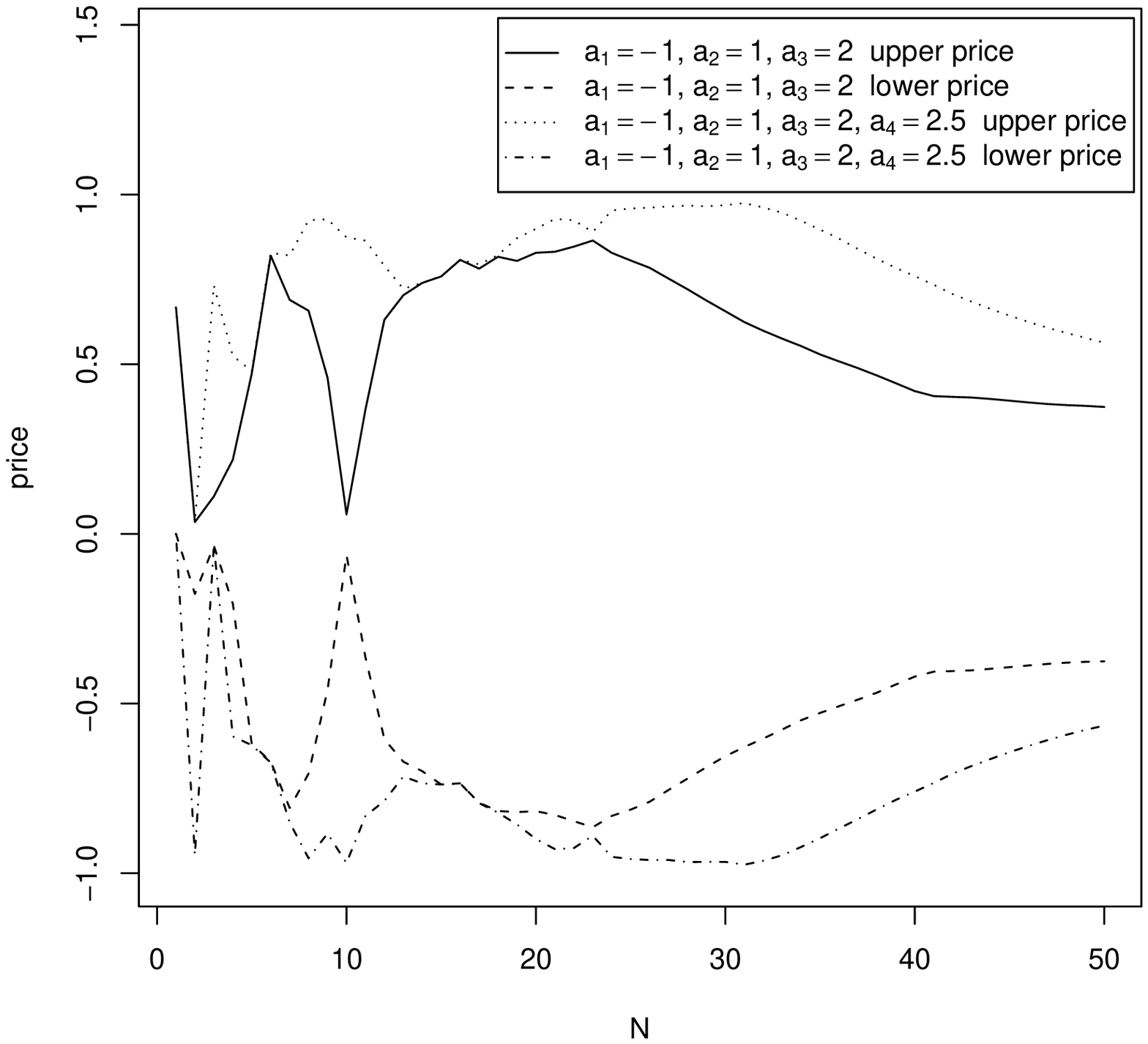}
\caption{$(\frac{0.5}{\sqrt{N}}=a_4 < a_2)$}
 \label{fig:Quad6}
\end{center}
\end{minipage}
\end{figure}

\begin{figure}[thbp]
  \centering  
  \includegraphics[width=11cm,keepaspectratio,clip]{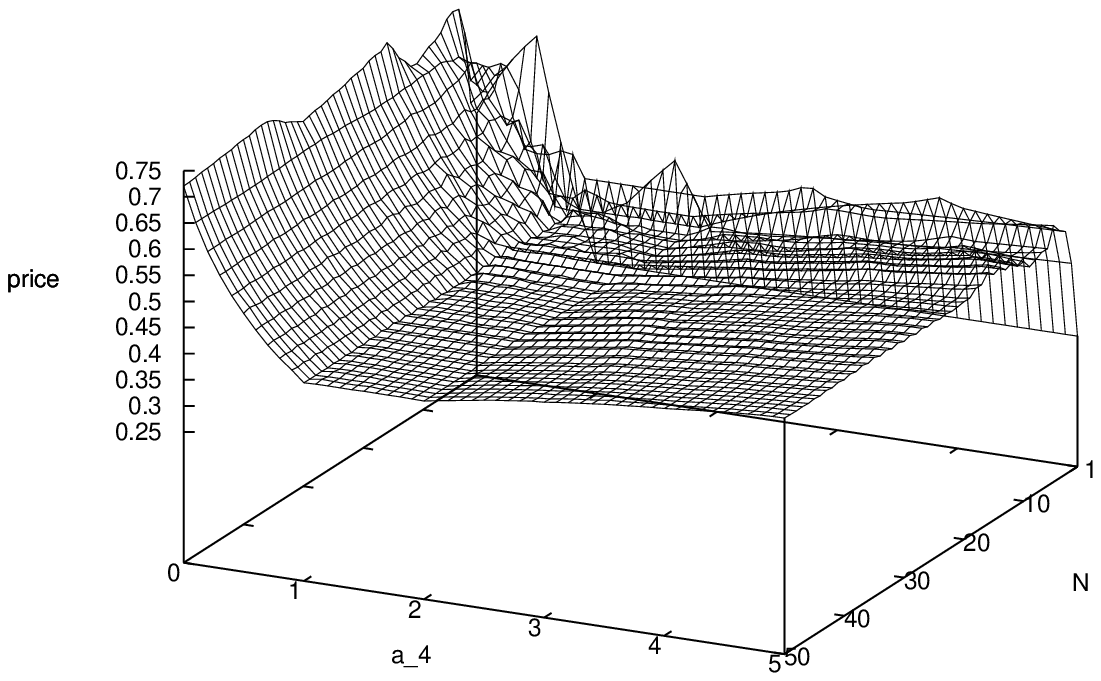}
  \caption{Comparison of trinomial model and quadnomial model for various $a_4$}
  \label{fig:Quad3d}
\end{figure}
Next, we vary the values of $a_4$ from 0 to 5 by $0.1$. Figure \ref{fig:Quad3d} displays the plot of the upper hedging prices of the butterfly spread for $1\le N\le 50$ and $0\le a_4 \le 5$.
From Figure \ref{fig:Quad3d}, we see that the upper hedging prices converge to an equal value in the interval $1 \le a_4 \le 2$.

Finally Figure  \ref{fig:pde3d} shows a numerical solution to the
partial differential equation (\ref{eq:main-pde}) for $0\le t\le 1$ and
$-2\le s \le 2$ for the case of Butterfly spread
$\payoff (s)=\max(0,S_N + 0.5) - 2 \max(0,S_N - 0.5) + \max(0,S_N - 1.5)$.
We compute an approximation of $\bar\phi(s,t)$ by the following difference scheme:
\begin{equation}
\label{eq:difference}
\frac{\bar \payoffc^{n+1}_i - \bar \payoffc^{n}_i}{\Delta t}=\frac{\tilde \sigma^2}{2}\frac{\bar \payoffc^{n}_{i+1} -2\bar \payoffc^{n}_{i} + \bar \payoffc^{n}_{i-1}}{\Delta s^2}, \qquad 
\begin{cases}
\tilde \sigma^2 = \bar\sigma^2, &  {\rm if}\  
\bar \payoffc^{n}_{i+1} -2\bar \payoffc^{n}_{i} + \bar \payoffc^{n}_{i-1}\ge 0 ,\\
\tilde \sigma^2=\underline\sigma^2,  &  {\rm if}\ \bar \payoffc^{n}_{i+1} -2\bar \payoffc^{n}_{i} + \bar \payoffc^{n}_{i-1}< 0.
\end{cases} 
\end{equation}
We rewrite (\ref{eq:difference}) as
\begin{equation}
\label{eq:difference2}
\bar \payoffc^{n+1}_i = \bar \payoffc^{n}_{i} + \frac{\tilde \sigma^2\Delta t}{2\Delta s^2}(\bar \payoffc^{n}_{i+1} -2\bar \payoffc^{n}_{i} + \bar \payoffc^{n}_{i-1}), \qquad 
\begin{cases}
\tilde \sigma^2 = \bar\sigma^2, &  {\rm if}\  
\bar \payoffc^{n}_{i+1} -2\bar \payoffc^{n}_{i} + \bar \payoffc^{n}_{i-1}\ge 0 ,\\
\tilde \sigma^2=\underline\sigma^2,  &  {\rm if}\ \bar \payoffc^{n}_{i+1} -2\bar \payoffc^{n}_{i} + \bar \payoffc^{n}_{i-1}< 0.
\end{cases} 
\end{equation}
We set $\underline\sigma^2 = 1$ and $\bar\sigma^2 = 2$ (\ref{eq:main-pde}).
For discretization we use $\Delta t = \frac{1}{300}$ and $\Delta s = \frac{1}{10}$,
which satisfies the stability condition 
(Section 8.4 of \cite{wilmott-howison-dewynne} or page 47 of \cite{smith-1985})
\[
\frac{\Delta t}{(\Delta s)^2} =\frac{1}{3} \le \frac{1}{2}
\]
for discretization of the heat equation
Since our partial differential equation (\ref{eq:main-pde}) can be understood
as a  piecewise heat equation, in our numerical experiments we found that 
$\Delta t$ and $\Delta s$ satisfying the same stability condition works well.
With $\Delta t = \frac{1}{300}$ and $\Delta s = \frac{1}{10}$,
we obtain $\bar \payoffc(0,1) \approx 0.3817$ and $\underline \payoffc(0,1) \approx 0.2060$ using difference scheme (\ref{eq:difference2}). 
In Figure \ref{fig:lower}, 
we compute (\ref{eq:main-pde2}) 
for the lower prices by the similar difference method.
Table \ref{tab:upper} shows the upper prices and the lower prices obtained in Figure \ref{fig:Butterfly2}. We find that these converge to the values obtained by the difference method for the partial differential equations (\ref{eq:main-pde}) and (\ref{eq:main-pde2}).


\begin{figure}[htbp]
\begin{minipage}{.5\linewidth}
\begin{center}
\includegraphics[width=8cm]{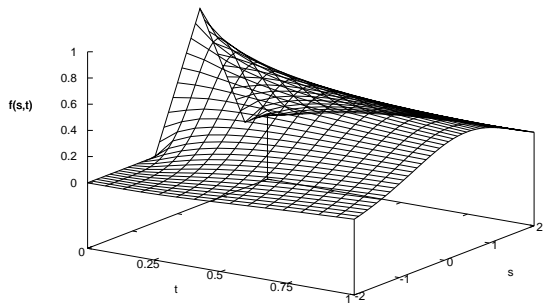}
 \caption{\footnotesize Numerical solution of the PDE \eqref{eq:main-pde}}
  \label{fig:pde3d}
\end{center}
\end{minipage}
\begin{minipage}{.5\linewidth}
\begin{center}
  \includegraphics[width=8cm]{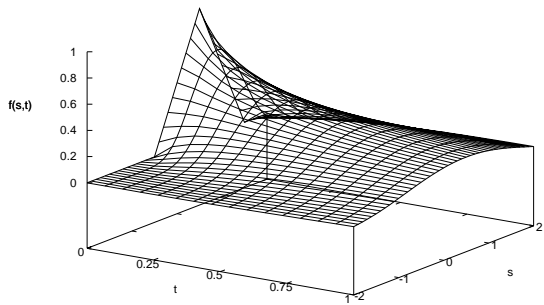}
  \caption{\footnotesize Numerical solution of the PDE \eqref{eq:main-pde2}}
  \label{fig:lower}
\end{center}
\end{minipage}
\end{figure}

\begin{table}[htbp]
\begin{center}
	\caption{Upper prices and lower prices of trinomial model}
	\begin{tabular}{|c|c|c|} \hline
	 $N$ & upper price & lower prices \\ \hline
	 1 & 0.2500 & 0.0000\\ \hline
	 20 & 0.3824 &0.1926  \\ \hline
	 40 & 0.3790& 0.1993 \\ \hline
	 60 & 0.3820 &0.2012 \\ \hline
	 80 & 0.3799 &0.2032  \\ \hline
	 100 & 0.3807 &0.2032
 \\ \hline
	\end{tabular}
	\label{tab:upper}
\end{center}
\end{table}

\section{Concluding remarks}
\label{sec:remarks}

In this paper we discussed various approximations and asymptotics
of upper hedging prices in multinomial models.  
In particular we showed that, as the number of rounds goes to infinity,
the upper hedging price of a European option converges 
to the solution of an additive form of the
Black-Scholes-Barenblatt equation.
By numerical experiments we checked that this convergence is fast
and the asymptotic approximation is useful.

A multinomial model is the simplest example of incomplete market.  A natural
extension of a multinomial model is the bounded forecasting game 
(\cite{shafer-vovk}), where Market's move is
a bounded interval containing the origin. This problem was already considered in
\cite{ruschendorf}.  Most results of this paper can be extended to the bounded
forecasting game.

Usually the Black-Scholes-Barenblatt equation is studied
in the case of vector-valued processes.  Then
the maximum variance and the minimum variance are no longer uniquely determined
and the maximization in each step of the game tree is more complicated.
Numerical studies of vector-valued cases are left to our future investigation.


\medskip
\noindent
{\bf Acknowledgments.}\quad  The authors are grateful to 
two reviewers for their careful reading and constructive comments.
We thank Takayasu Matsuo
for very useful comments on numerical solution of the
partial differential equation and Hitoshi Ishii for 
his very kind explanation on existence of the classical solution 
or the viscosity solution of the partial differential equation.

\end{document}